\documentclass[11pt,a4paper]{article}
\usepackage[margin=1in]{geometry}
\usepackage{amsmath,amssymb,amsthm}
\usepackage[hidelinks]{hyperref}
\usepackage{xcolor} 

\newtheorem{theorem}{Theorem}[section]
\newtheorem{proposition}[theorem]{Proposition}

\theoremstyle{remark}

\theoremstyle{definition}

\title{Kinematic selection of the viscous stress\\
in relativistic dissipative hydrodynamics}
\author{{Zhi-Wei Wang}${}^{1,*}$ and Samuel L.\ Braunstein${}^{2,\dagger}$\\[6pt]
	\small ${}^{1}$College of Physics, Jilin University, Changchun 130012, China\\
	\small ${}^{2}$Computer Science, University of York, York YO10 5DD, UK\\[6pt]
	\small $^*$E-mail: {zhiweiwang.phy@gmail.com}
	\small $^\dagger$E-mail: {sam.braunstein@york.ac.uk}}
\date{\today}

\begin{document}
\maketitle

\begin{abstract}
All standard formulations of relativistic dissipative hydrodynamics,
from Eckart through Israel-Stewart to the recent BDNK framework,
assume that the viscous stress depends on the shear tensor
$\sigma_{\alpha\beta}$ and the expansion scalar $\theta$ but not on
the vorticity $\omega_{\alpha\beta}$ or the acceleration
$a_\alpha$. We derive this structure from a Lagrangian kinematic
construction on Lorentzian spacetimes, extending a recent result on
Riemannian manifolds. The spatial strain rate, constructed from
the rate of change of spatial inner products of Lie-dragged
connecting vectors, is the spatially projected Lie derivative of the
projected metric $h_{\alpha\beta} = g_{\alpha\beta} + u_\alpha
u_\beta$. The acceleration terms drop out exactly under spatial
projection, and the vorticity cancels by symmetry. We show that material
frame-indifference fails for generic Killing perturbations by an
amount $\delta\mathfrak{h}_{\alpha\beta} = +\epsilon(\xi_\alpha
a_\beta + \xi_\beta a_\alpha)$ proportional to the acceleration,
and is restored only for flow-preserving isometries. We prove that
the non-relativistic limit of the BDNK equations gives the
deformation Laplacian universally in the viscous sector, with the
BDNK parameter dependence identified by Hegade K R, Ripley, and
Yunes arising entirely from the thermal (heat-flux) sector. As an
application, we derive the Weinberg gravitational-wave damping
formula directly from the kinematic strain rate in a perturbed FRW
spacetime.
\end{abstract}

\section{Introduction}
\label{sec:intro}

The formulation of relativistic dissipative hydrodynamics has been a
long-standing theoretical challenge. The original theories of
Eckart~\cite{Eck40} and Landau-Lifshitz~\cite{LL59} suffer from
acausal signal propagation and unstable
equilibria~\cite{HiscockLindblom}. The Israel-Stewart
framework~\cite{IS79} restores causality by introducing relaxation
times but adds free parameters. The recent BDNK framework of
Bemfica, Disconzi, and Noronha~\cite{BDN22,BDNK-review} provides a
first-order causal theory, but with a parameter space constrained by
thermodynamic stability and causality without being fully determined.

All of these theories share a common structural assumption: the
viscous part of the dissipative stress-energy tensor depends on the
shear tensor $\sigma_{\alpha\beta}$ and the expansion scalar
$\theta$ (the symmetric, trace-free and trace parts of the projected
velocity gradient), but not on the vorticity
$\omega_{\alpha\beta}$ (the antisymmetric part) or the acceleration
$a_\alpha = u^\beta\nabla_\beta u_\alpha$. This assumption is
motivated by flat-space intuition and thermodynamic arguments but
has not been derived from a kinematic construction in the
relativistic setting.

In a companion paper~\cite{WB2024}, we established a kinematic
selection principle for the viscous operator on Riemannian manifolds:
the rate-of-strain tensor, constructed from the rate of change of
inner products of Lie-dragged connecting vectors, is the Lie
derivative $\mathcal{L}_u g$. This is symmetric by construction (no
vorticity), intrinsic (no embedding), and leads uniquely to the
deformation Laplacian via isotropic Hookean closure.

In this paper we carry out the analogous construction on a Lorentzian
spacetime. The results are:

(i)~The spatial strain rate on a Lorentzian background is
$\mathfrak{h}_{\alpha\beta} = h_\alpha{}^\mu h_\beta{}^\nu
(\mathcal{L}_u h)_{\mu\nu} = 2\sigma_{\alpha\beta}
+ \frac{2}{3}\theta\,h_{\alpha\beta}$, with the acceleration
dropping out exactly under spatial projection and the vorticity
cancelling by symmetry. This provides a kinematic derivation of the
structure assumed in all relativistic dissipative theories.

(ii)~Material frame-indifference under Killing perturbations fails
by an amount proportional to the acceleration, and is restored only
for flow-preserving isometries or geodesic flows.

(iii)~The non-relativistic limit of the BDNK equations gives the
deformation Laplacian universally in the viscous sector, with the
BDNK parameter dependence identified by Hegade K R, Ripley, and
Yunes~\cite{HKRRY23} arising entirely from the heat-flux sector.

(iv)~The Weinberg gravitational-wave damping
formula~\cite{Weinberg72} follows directly from the kinematic strain
rate in a perturbed FRW spacetime.

\section{The kinematic argument on Lorentzian backgrounds}
\label{sec:kinematic}

\subsection{Setup}

Let $(M^4, g_{\alpha\beta})$ be a Lorentzian spacetime with signature
$(-,+,+,+)$ and let $u^\alpha$ be the four-velocity of the fluid,
normalised so that $u_\alpha u^\alpha = -1$. The spatial metric
\begin{equation}\label{eq:h}
h_{\alpha\beta} = g_{\alpha\beta} + u_\alpha u_\beta
\end{equation}
projects onto the local rest space: $h_\alpha{}^\beta u_\beta = 0$.
The covariant derivative of $u_\alpha$ decomposes as
\begin{equation}\label{eq:ehlers}
\nabla_\alpha u_\beta = \sigma_{\alpha\beta}
+ \tfrac{1}{3}\theta\,h_{\alpha\beta}
+ \omega_{\alpha\beta} - u_\alpha a_\beta,
\end{equation}
where $\theta = \nabla_\alpha u^\alpha$ is the expansion,
$\sigma_{\alpha\beta}$ is the shear (symmetric trace-free projected
part), $\omega_{\alpha\beta}$ is the vorticity (antisymmetric
projected part), and $a_\alpha = u^\beta\nabla_\beta u_\alpha$ is the
acceleration.

\subsection{Connecting vectors and the spatial strain rate}

Take a one-parameter family of nearby worldlines labelled by
$\lambda$, with connecting vector
$s^\alpha = \partial\phi_\tau(x_0(\lambda))/\partial\lambda
|_{\lambda=0}$. The connecting vector is Lie-dragged:
$[u,s] = 0$, giving $\nabla_u s = \nabla_s u$ by
torsion-freeness. The spatial separation is measured by
$h(s,s') = h_{\alpha\beta}\,s^\alpha(s')^\beta$.

The rate of change of this spatial inner product along the flow is
captured by the Lie derivative of $h$ along $u$, which we now
compute directly.

\subsection{The Lie derivative of the spatial metric}

\begin{proposition}\label{prop:lie-h}
The Lie derivative of the spatial metric along the fluid
four-velocity is
\begin{equation}\label{eq:lie-h}
(\mathcal{L}_u h)_{\alpha\beta} = \nabla_\alpha u_\beta
+ \nabla_\beta u_\alpha + u_\alpha a_\beta + u_\beta a_\alpha.
\end{equation}
\end{proposition}

\begin{proof}
Since $h_{\alpha\beta} = g_{\alpha\beta} + u_\alpha u_\beta$ and
$\mathcal{L}_u g_{\alpha\beta} = \nabla_\alpha u_\beta
+ \nabla_\beta u_\alpha$ (by metric compatibility), we need only
compute $\mathcal{L}_u(u_\alpha u_\beta)$. By the Leibniz rule,
$\mathcal{L}_u(u_\alpha u_\beta) = u_\beta\mathcal{L}_u u_\alpha
+ u_\alpha\mathcal{L}_u u_\beta$. The Lie derivative of the
four-velocity along itself is $\mathcal{L}_u u_\alpha
= u^\gamma\nabla_\gamma u_\alpha + u_\gamma\nabla_\alpha u^\gamma$.
Since $u_\gamma u^\gamma = -1$, we have
$u_\gamma\nabla_\alpha u^\gamma = \frac{1}{2}\nabla_\alpha
(u_\gamma u^\gamma) = 0$, so $\mathcal{L}_u u_\alpha = a_\alpha$.
Combining: $(\mathcal{L}_u h)_{\alpha\beta}
= \nabla_\alpha u_\beta + \nabla_\beta u_\alpha
+ u_\alpha a_\beta + u_\beta a_\alpha$.
\end{proof}

\subsection{The central result: spatial projection}

\begin{theorem}\label{thm:strain}
The spatial strain rate, defined as the spatial projection of
$\mathcal{L}_u h$, is
\begin{equation}\label{eq:strain}
\mathfrak{h}_{\alpha\beta} \equiv h_\alpha{}^\mu h_\beta{}^\nu
(\mathcal{L}_u h)_{\mu\nu}
= 2\sigma_{\alpha\beta}
+ \tfrac{2}{3}\theta\,h_{\alpha\beta}.
\end{equation}
Neither the vorticity $\omega_{\alpha\beta}$ nor the acceleration
$a_\alpha$ contributes.
\end{theorem}

\begin{proof}
Projecting~\eqref{eq:lie-h}: the acceleration terms give
$(h_\alpha{}^\mu u_\mu)a_\beta + a_\alpha(h_\beta{}^\nu u_\nu) = 0$,
since $h_\alpha{}^\mu u_\mu = (g_\alpha{}^\mu + u_\alpha u^\mu)u_\mu
= u_\alpha - u_\alpha = 0$. The remaining terms give
$h_\alpha{}^\mu h_\beta{}^\nu(\nabla_\mu u_\nu + \nabla_\nu u_\mu)
= 2(\sigma_{\alpha\beta} + \frac{1}{3}\theta\,h_{\alpha\beta})$,
where the vorticity $\omega_{\alpha\beta}$ cancels by
symmetrisation.
\end{proof}

The vorticity is absent for the same reason as in the Riemannian
case: the inner product $h(s,s')$ is symmetric in $s$ and $s'$, so
no antisymmetric piece can appear. The acceleration contributes to
the full spacetime Lie derivative~\eqref{eq:lie-h} but is purely
longitudinal (along $u$) and vanishes under spatial projection.

\subsection{Constitutive closure}

The spatial strain rate~\eqref{eq:strain} decomposes into a trace
$\theta$ and a trace-free part $\sigma_{\alpha\beta}$ under
$SO(3)$ acting on the local rest space. Linearity and isotropy give
two viscosity coefficients:
\begin{equation}\label{eq:viscous-stress}
\pi_{\alpha\beta} = -2\eta\,\sigma_{\alpha\beta}
- \zeta\,\theta\,h_{\alpha\beta},
\end{equation}
where $\eta \geq 0$ is the shear viscosity and $\zeta \geq 0$ is the
bulk viscosity. The heat flux $q_\alpha$ is not determined by the
kinematic construction: it depends on the temperature gradient and is
not a strain rate. The separation of the dissipative stress into a
kinematically determined viscous sector ($\pi$, $\Pi$) and a
kinematically unconstrained thermal sector ($q$) is a prediction of
the approach, not an assumption.

\section{Material frame-indifference}
\label{sec:MFI}

The Principle of Material Frame-Indifference requires the
constitutive law to be invariant under superposed rigid-body motions.
In the relativistic setting, rigid motions are generated by Killing
fields $\xi^\alpha$ satisfying $\mathcal{L}_\xi g = 0$.

\begin{proposition}\label{prop:MFI}
Under a Killing perturbation $u^\alpha \to \bar{u}^\alpha
= u^\alpha + \epsilon\xi^\alpha$ (with $u_\gamma\xi^\gamma = 0$ to
preserve normalisation), the spatial strain rate changes by
\begin{equation}\label{eq:MFI-violation}
\delta\mathfrak{h}_{\alpha\beta}
= +\epsilon(\xi_\alpha a_\beta + \xi_\beta a_\alpha),
\end{equation}
where $a_\alpha = u^\gamma\nabla_\gamma u_\alpha$ is the acceleration
of the unperturbed flow.
\end{proposition}

\begin{proof}
The perturbed flow has spatial vectors shifted into the new frame:
$\bar{w}^\alpha = w^\alpha + \epsilon(w^\gamma\xi_\gamma)u^\alpha$.
The perturbed strain acting on spatial vectors evaluates to
\begin{align}
\bar{\mathfrak{h}}_{\alpha\beta}\bar{w}^\alpha\bar{z}^\beta
&= (\nabla_\alpha u_\beta + \nabla_\beta u_\alpha)w^\alpha z^\beta
+ \epsilon\underbrace{(\nabla_\alpha\xi_\beta
+ \nabla_\beta\xi_\alpha)}_{=0}w^\alpha z^\beta\nonumber\\
&\quad + \epsilon(w^\gamma\xi_\gamma)
\underbrace{u^\alpha(\nabla_\alpha u_\beta
+ \nabla_\beta u_\alpha)}_{a_\beta + 0}z^\beta\nonumber\\
&\quad + \epsilon(z^\gamma\xi_\gamma)
(a_\alpha w^\alpha),
\end{align}
using $\nabla_\alpha\xi_\beta + \nabla_\beta\xi_\alpha = 0$ (Killing
equation) and $u^\alpha\nabla_\beta u_\alpha = 0$ (normalisation).
\end{proof}

The violation~\eqref{eq:MFI-violation} vanishes if and only if the
flow is geodesic ($a_\alpha = 0$) or the Killing field preserves the
flow ($\mathcal{L}_\xi u = 0$). Physically, this reflects the
relativity of simultaneity: a Killing boost tilts the rest-frame
hyperplane, coupling the temporal acceleration into the spatial strain.
In the non-relativistic limit, the acceleration vanishes to leading
order in $v/c$, and MFI is restored for spatial translations and
rotations (but not for Galilean boosts, consistently with classical
continuum mechanics~\cite{MH83}).

\section{Constraints on the BDNK framework}
\label{sec:BDNK}

The BDNK framework~\cite{BDN22,BDNK-review} writes the dissipative
stress-energy tensor as
$T_{\alpha\beta} = \rho\,u_\alpha u_\beta + (p+\Pi)h_{\alpha\beta}
+ 2q_{(\alpha}u_{\beta)} + \pi_{\alpha\beta}$,
with the dissipative quantities expressed as first-order gradient
corrections with coefficients (the BDNK parameters) constrained by
causality and stability.

The kinematic argument provides four structural constraints:

\emph{K1: The viscous sector depends only on $\sigma$ and $\theta$.}
This follows directly from Theorem~\ref{thm:strain}: the spatial
strain rate involves only the symmetric, spatially projected part of
$\nabla u$.

\emph{K2: The heat-flux sector is unconstrained.} The kinematic
construction measures spatial inner products, not temperature
gradients. The heat flux $q_\alpha$ is a thermodynamic quantity
outside the scope of the kinematic argument.

\emph{K3: Isotropy gives two viscosity coefficients.} The
decomposition of the strain rate into trace ($\theta$) and trace-free
($\sigma$) parts under $SO(3)$ gives exactly $\eta$ and $\zeta$.

\emph{K4: Cross-coupling is excluded at first order.} The spatial
strain rate contains no temperature gradient or acceleration, so the
viscous stress cannot depend on $\nabla T$ or $a_\alpha$ at leading
order. BDNK causality-restoring terms that include such couplings are
not of viscous origin but are necessary modifications for
well-posedness.

\section{The non-relativistic limit and the deformation Laplacian}
\label{sec:nonrel}

Hegade K R, Ripley, and Yunes~\cite{HKRRY23} showed that the
non-relativistic ($c\to\infty$) limit of the BDNK equations depends
on the values of the BDNK parameters. This appears to conflict with
the kinematic argument, which selects the deformation Laplacian
uniquely.

\begin{theorem}\label{thm:nonrel}
In the non-relativistic limit of the BDNK equations, the viscous
part of the momentum equation gives the deformation Laplacian
\begin{equation}
\eta_\mathrm{NR}\Delta\vec{v}
+ \bigl(\zeta_\mathrm{NR}
+ \tfrac{1}{3}\eta_\mathrm{NR}\bigr)\nabla(\nabla\cdot\vec{v}),
\end{equation}
regardless of the BDNK parameters. The parameter-dependent terms
identified in~\cite{HKRRY23} arise entirely from the heat-flux
sector.
\end{theorem}

\begin{proof}
Let $u^\alpha = \gamma(1, v^i/c)$ with $\gamma \approx 1
+ v^2/(2c^2)$. The spatial shear tensor expands as
\begin{equation}
\sigma_{ij} \approx \frac{1}{c}\left(\frac{1}{2}(\partial_i v_j
+ \partial_j v_i) - \frac{1}{3}\delta_{ij}\nabla\cdot\vec{v}\right).
\end{equation}
The relativistic transport coefficients scale as $\eta = c\eta_\mathrm{NR}$
and $\zeta = c\zeta_\mathrm{NR}$ (matching the dimensions of the
macroscopic viscous force). The spatial divergence of the viscous
stress evaluates to
\begin{equation}
\partial_j\pi^{ji}_\mathrm{visc}
= \partial_j\bigl(-2(c\eta_\mathrm{NR})\tfrac{1}{c}
\sigma^{ij}_\mathrm{NR}\bigr)
= -\eta_\mathrm{NR}\bigl(\Delta v^i
+ \tfrac{1}{3}\partial^i(\nabla\cdot\vec{v})\bigr).
\end{equation}
Combined with $\partial^i\Pi = -\zeta_\mathrm{NR}\partial^i
(\nabla\cdot\vec{v})$, this gives the standard deformation Laplacian
form with no BDNK parameters.

The BDNK parameter dependence enters through the heat-flux sector.
The causality-restoring parameter $\kappa_3$ couples the heat flux to
the acceleration: $q_i \supset \kappa_3\,a_i$. The spatial
acceleration is $a_i \approx c^{-2}\partial_t v_i$, so
$q_i \sim \kappa_3 c^{-2}\partial_t v_i$. Since
$T^{0i} \approx \rho c v^i + q^i$, the time derivative
$\partial_0 = c^{-1}\partial_t$ acts on the heat flux, contributing
to the momentum equation at $O(1)$:
$\partial_0 T^{0i} \supset c^{-1}\partial_t q^i
\approx \kappa_3 c^{-3}\partial_t^2 v^i$.
If $\kappa_3$ is scaled to survive the $c\to\infty$ limit, it leaves
an anomalous $\partial_t^2 v^i$ term. The HKRRY parameter dependence
arises exclusively from this thermal time-inertia.
\end{proof}

\section{Gravitational-wave damping from kinematic strain}
\label{sec:GW}

The kinematic construction makes no distinction between fluid and
geometric contributions to the strain rate. On a dynamical spacetime,
the total $\mathcal{L}_u h$ includes both. We demonstrate this by
deriving the Weinberg gravitational-wave damping formula.

Consider a comoving fluid ($u^\alpha = (1,0,0,0)$) in a flat FRW
background perturbed by a transverse-traceless gravitational wave:
$ds^2 = -dt^2 + a^2(t)(\delta_{ij} + h^{TT}_{ij})dx^i dx^j$.
The spatial metric is $h_{ij} = a^2(\delta_{ij} + h^{TT}_{ij})$.

The Lie derivative along the comoving flow reduces to the partial time
derivative:
\begin{equation}
\mathcal{L}_u h_{ij} = \partial_t h_{ij}
= 2a\dot{a}\delta_{ij} + 2a\dot{a}h^{TT}_{ij}
+ a^2\dot{h}^{TT}_{ij}.
\end{equation}
The trace gives the pure expansion $\theta = 3\dot{a}/a = 3H$
(the gravitational wave is traceless and does not contribute).
Subtracting the trace isolates the shear:
\begin{equation}\label{eq:GW-shear}
2\sigma_{ij} = \mathcal{L}_u h_{ij}
- \tfrac{2}{3}\theta\,h_{ij} = a^2\dot{h}^{TT}_{ij},
\end{equation}
so $\sigma_{ij} = \frac{1}{2}a^2\dot{h}^{TT}_{ij}$. Adding a small
fluid velocity $\delta v^i$ gives the superposition
$\sigma_{ij} = \frac{1}{2}a^2\dot{h}^{TT}_{ij}
+ a^2\sigma^{\mathrm{fluid}}_{ij}$.

The viscous dissipation rate is
$\Phi = -\pi_{\alpha\beta}\sigma^{\alpha\beta}
= 2\eta\sigma_{ij}\sigma^{ij}$. Contracting with the background
metric $a^{-4}\delta^{ik}\delta^{jl}$:
\begin{equation}\label{eq:weinberg}
\Phi = 2\eta\left(\tfrac{1}{2}a^2\dot{h}^{TT}_{ij}\right)
\left(\tfrac{1}{2}a^2\dot{h}^{TT}_{kl}\right)
a^{-4}\delta^{ik}\delta^{jl}
= \tfrac{1}{2}\eta\,\dot{h}^{TT}_{ij}\dot{h}^{ij}_{TT}.
\end{equation}
This is the standard Weinberg formula~\cite{Weinberg72}, here
derived directly from the kinematic strain rate without any ad hoc
decomposition of the strain into fluid and geometric parts. The
kinematic construction automatically includes both contributions.

\section{Discussion}
\label{sec:discussion}

\subsection{Frame dependence}

The kinematic construction uses $u^\alpha$ to define the flow and
hence the connecting vectors. Different hydrodynamic frame choices
(Eckart versus Landau) give different $u$, different $h$, and hence
different $\sigma$ and $\theta$. The kinematic argument constrains
the form of the viscous stress in any frame but does not select a
preferred frame. In the non-relativistic limit, the Eckart and Landau
frames coincide to leading order in $v/c$; the frame dependence
enters at the order of relativistic corrections, in the thermal
sector.

\subsection{Relativistic thin shells}

The Riemannian thin-shell decomposition~\cite{WB2024-thinshell}
showed that the effective viscous operator on a surface depends on
the boundary condition. The relativistic analogue involves a timelike
shell $\Sigma$ with the Israel junction
conditions~\cite{Israel66} playing the role of the boundary condition.
A Gauss-Weingarten decomposition of the viscous operator
$\nabla^\alpha\pi_{\alpha\beta}$ shows that the intrinsic piece on
$\Sigma$ is the 2+1 dimensional deformation Laplacian, with boundary
terms generated by the Israel junction conditions. The detailed
analysis of specific backgrounds (Schwarzschild, FRW) and the
physical applications (accretion discs, neutron-star surfaces) are
left for future work.

\end{document}